\documentclass{llncs}
\usepackage{amsmath,amssymb,amsfonts,stmaryrd}
\usepackage{graphicx}
\usepackage{xcolor}
\usepackage{hyperref}
\usepackage[utf8]{inputenc}
\usepackage[T1]{fontenc}
\usepackage{booktabs}
\usepackage{xspace}

\usepackage{algorithm}
\usepackage[noend]{algpseudocode}
\usepackage{siunitx}

\newif\ifcomments%
\commentstrue%

\newcommand{\comment}[1]{}

\newcommand{\R}{\ensuremath{\mathbb{R}}}

\begin{document}

%\title{A Polynomialization Algorithm for Elementary ODEs with Applications to CRN Design}
\title{Compiling Elementary Mathematical Functions\\ into Finite Chemical Reaction Networks\\
  via a Polynomialization Algorithm for ODEs}
\author{Mathieu Hemery \and Fran\c{c}ois Fages \and Sylvain Soliman}
\institute{Inria Saclay Île-de-France, EPI Lifeware, Palaiseau, France}
\maketitle
\thispagestyle{plain}

\begin{abstract}
  The Turing completeness result for continuous chemical reaction networks (CRN) shows that any computable function over the real numbers
  can be computed by a CRN over a finite set of formal molecular species using at most bimolecular reactions with mass action law kinetics.
  The proof uses a previous result of Turing completeness for functions defined by polynomial ordinary differential equations (PODE),
  the dual-rail encoding of real variables by the difference of concentration between two molecular species,
  and a back-end quadratization transformation to restrict to elementary reactions with at most two reactants.
  In this paper, we present a polynomialization algorithm of quadratic time complexity to transform a system of elementary differential equations to PODE\@.
  This algorithm is used as a front-end transformation to compile any elementary mathematical function,
  either of time or of some input species, into a finite CRN\@.
  We illustrate the performance of our compiler on a benchmark of elementary functions
  relevant to CRN design problems in synthetic biology specified by mathematical functions.
  In particular, the abstract CRN obtained by compilation of the Hill function of order $5$ is
  compared to the natural CRN structure of MAPK signalling networks.
\end{abstract}

\section{Introduction}

Chemical reaction networks (CRN) provide a standard formalism in chemistry and biology to describe, analyze, and also design
complex molecular interaction networks.
In the perspective of systems biology, they are a central tool to analyze the high-level functions of the cell
in terms of their low-level molecular interactions.
In the perspective of synthetic biology, they constitute a target programming language to implement in chemistry new functions
in either living cells or artificial devices.

A CRN can be interpreted in a hierarchy of Boolean, discrete, stochastic and differential semantics~\cite{CSWB09ab,FS08tcs}
which is at the basis of a rich theory for the analysis of their dynamical properties~\cite{Feinberg77crt,CF06siamjam,BFS18jtb},
and more recently, of their computational power~\cite{CSWB09ab,CDS12nc,FLBP17cmsb}.
In particular, their interpretation by Ordinary Differential Equations (ODE)
allows us to give a precise mathematical meaning to the notion of analog computation and high-level functions computed by cells~\cite{DRSL13nature,SK13nature,RRD16bbs},
using the following definitions:

\begin{definition}\cite{FLBP17cmsb,GC03,Shannon41}\label{generates}
  A function $f: \mathbb{R_+}  \to \mathbb{R_+}$ is \emph{generated by a CRN}
  on some species $y$ with given initial concentrations for all species,
  if the ODE associated to the CRN has a unique solution verifying $\forall t\ge 0\ y(t)=f(t)$.
\end{definition}

That first definition states that a positive real function of one positive argument is \emph{generated} by a CRN for some given initial concentration values,
if the graph of that function is given by the temporal evolution of the concentration of one molecular species in that CRN.

\begin{definition}\cite{FLBP17cmsb,GC03}\label{computes}
  A function $f: \mathbb{R_+}  \to \mathbb{R_+}$ is \emph{computed by a CRN}
  for some input species $x$, output species $y$, and initial concentrations given for all species apart from $x$,
  if for any input concentration value $x(0)$ for $x$,
  the ODE initial value problem associated to the CRN has a unique solution
  satisfying\footnote{For the sake of simplicity of the definition given here, we omit the error control mechanism that requires one extra CRN species $z$ verifying:\\ $\forall t>1\ |y(t)-f(x(0))|\le z(t),\ \forall t'>t\ z(t')<z(t)$ and $\lim_{t\rightarrow\infty}z(t)=0$.} $\lim_{t\rightarrow\infty}y(t)=f(x(0))$. %where the output species $y$ converges to $f(x(0))$.
\end{definition}

The second definition states that the same function is \emph{computed} by a CRN if for any input $x\ge 0$,
and initialization of the CRN input species to value $x$,
the CRN output species converges to the result $f(x)$.
That definition for input/output functions computed by a CRN
is used in~\cite{FLBP17cmsb} to show
the Turing completeness of continuous CRNs in the sense that any computable function over the real numbers
  can be computed by a CRN over a finite set of formal molecular species using at most bimolecular reactions with mass action law kinetics.
  The proof uses a previous result of Turing completeness for functions defined by polynomial ordinary differential equation initial value problems (PIVP)~\cite{BCGH06complexity},
  the dual-rail encoding of real variables by the difference of concentration between two molecular species~\cite{OK11iet,HT79cmsjb},
  and a back-end quadratization transformation to restrict to elementary reactions with at most two reactants~\cite{CPSW05ejde,HFS20cmsb,BP21iwoca}.
  This proof gives rise to a pipeline, implemented in BIOCHAM-4\footnote{\url{http://lifeware.inria.fr/biocham/}. All experiments described in this paper are available at \url{https://lifeware.inria.fr/wiki/Main/Software\#CMSB21}}, to compile any computable real function presented by a PIVP into a finite CRN\@.

However in practice, it is not immediate to define a PIVP that generates or computes a desired function.
In this article, we solve this problem for elementary functions over the reals,
by adding to our compilation pipeline a front-end module to transform any elementary function to a PIVP
which either generates or computes that function, as schematized in Fig.~\ref{fig:pipeline}.

\begin{figure}[h]
\begin{center}
	\includegraphics[width=\textwidth]{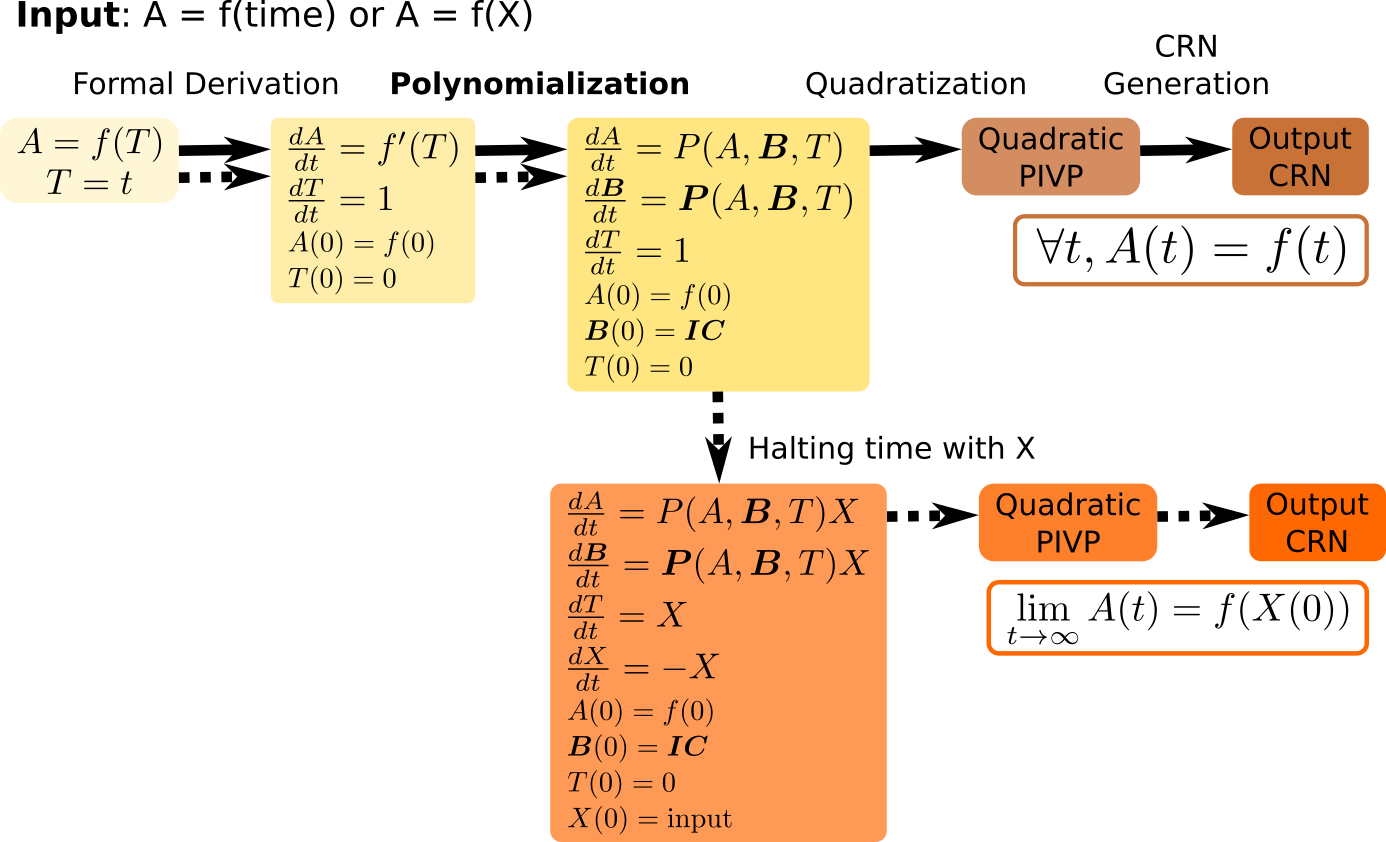}
\end{center}
\caption{Polynomialization step in the complete pipeline for compiling a formally differentiable function $f$ (termination is proved here for elementary functions) into a finite CRN,
  either a function of time (plain arrows) or an input/output function (dashed arrows):
        $P$ and $\boldsymbol{P}$ are polynomials, and $\boldsymbol{B}$ denotes the set of species introduced by
	polynomialization given with initial conditions $\boldsymbol{IC}$.}\label{fig:pipeline}
\end{figure}

More precisely, we present a polynomialization algorithm to transform
any Elementary ODE system (EODE), i.e.~ODE system in explicit form made of elementary differential functions, to a polynomial one (PODE).
This algorithm proceeds by introducing variables
for computing the non-polynomial terms of the input, eliminating such terms from the ODE by rewriting,
and obtaining the ODE for the new variables by formal derivation.
The derivation steps may bring new non-polynomial terms requiring new variables.
We show the termination of that algorithm,
with quadratic time complexity, and that only a linear number of new variables,
in terms of the size of the input expression, are actually needed.

Related work includes one method described in~\cite{S12arXiv} to compute some polynomial abstractions of a non-polynomial system.
That method progresses top-down from some set of possible functions given a priori,
and progressively sieves them down to a polynomial system.
This permits to choose the degree of the polynomial abstraction, but that method may fail
if eitherr the starting set of functions or the chosen degree are too small.
On the other hand, our algorithm proceeds bottom-up, by introducing the functions when needed, with no choice a priori.
A bottom-up approach closer to ours is also mentioned in~\cite{G11IEEE} but for a very restricted grammar of functions while we develop here a
general algorithm and prove its termination on the whole class of elementary mathematical functions.
One can also cite the function \verb$dpolyform$ of Maple PDETools package which returns
  a PODE in implicit form for which the expression given as input is guaranteed
to be a solution, but that does not provide a polynomial expression for the derivative of
each variable, i.e.~a PODE in explicit form, as required for our compilation pipeline.

The rest of the paper is organized as follows.
In the next section, we describe the language of elementary function that are accepted as input of our compilation pipeline into CRN\@.
In Sec.~\ref{algo}, we present a general polynomialization algorithm for elementary ODEs,
prove its termination, and show its quadratic time complexity.
In Sec.~\ref{pipeline} we describe the use of this algorithm as a front-end transformation in our compilation pipeline
to compile any elementary mathematical function into a finite CRN\@.
In Sec.~\ref{eval} we evaluate this approach on a benchmark of elementary functions relevant to CRN design problems in synthetic biology used in~\cite{HFS20cmsb}.
In particular, we compare the CRN synthesized for the Hill function of order 5 to the structure of the natural MAPK signalling CRNs
that have been shown to implement a similar input/output function~\cite{HF96pnas}.
Finally, we conclude on the results achieved and several perspectives for future work.

\section{Input Language of Elementary Functions}

\subsection{Example}

Let us consider the problem of synthesizing a CRN to \emph{generate} the function of time: $A(t) = \log(1+t^2)$ in the sense of Def.~\ref{generates}.
The compilation method described in~\cite{FLBP17cmsb} takes as input a PIVP which admits that function as solution on variable $A$.
Here, we want to automate the construction of such a PIVP\@.

The first step of our front-end transformation is schematized in Fig.~\ref{fig:pipeline}, is
to determine an ODE which admits that function of time as solution.
For this, one can simply take the derivative of the equation with respect to time and set as initial
condition the value for the desired function at $0$, giving:
\[
	\frac{dA}{dt} = \frac{2t}{1+t^2} \qquad A(0) = 1
\]
Then we need to transform this ODE to a PODE\@.
Our polynomialization algorithm will introduce a new
variable $B = \frac{1}{1+t^2}$, and similarly compute its derivative and its initial value, as follows:
\[
	\frac{dB}{dt} = \frac{-2t}{{(1+t^2)}^2} = -2tB^2 \qquad B(0)=1
\]
We also need to introduce a variable $T$ for time with $\frac{dT}{dt} = 1$ and $T(0)=0$.
giving the following PIVP\@:
\begin{align*}
	\frac{dA}{dt} &= 2.T.B &\frac{dB}{dt} &= -2.T.B^2 &\frac{dT}{dt} &= 1\\
	A(0)&=1& B(0)&= 1&T(0)&=0
\end{align*}
Note that the termination of those transformations is not obvious in general. It is proved
in the next section.
That PIVP of degree $3$ can now be used as input of our previous compilation
pipeline~\cite{FLBP17cmsb}. It is first transformed to a quadratic form~\cite{HFS20cmsb}, in this case by introducing one
variable, $BT = B.T$, and removing the time $T$, giving the following quadratic PIVP:
\begin{align*}
	\frac{dA}{dt} &= 2.BT &\frac{dB}{dt} &= -2.BT.B & \frac{d(BT)}{dt} &= \frac{dB}{dt}.T +
	B\frac{dT}{t} =-2.BT^2 + B\\
	A(0)&=1& B(0)&= 1&BT(0)&=0
\end{align*}

One reaction with mass action law kinetics is then generated for each monomial of the ODE.
Since in this example the reactions are well-formed and strict the system is positive (lemma 1 in \cite{FGS15tcs}).
There is thus no need to introduce dual-rail variables for negative values,
and the generated elementary CRN (with rate constants written above the arrow) is:
\begin{align*}
B+\mathit{BT}&\xrightarrow{2} \mathit{BT}&    B&\xrightarrow{1}B+\mathit{BT}\\
2.\mathit{BT}&\xrightarrow{2}\mathit{BT}&    \mathit{BT}&\xrightarrow{2}A+\mathit{BT}&B(0) &= 1
\end{align*}

Now, it is worth noting that if we want to synthesize a CRN that \emph{computes} (instead of generating)
the function $\log(1+x^2)$ of some input $x$ in the sense of Def.~\ref{computes},
the general method described in~\cite{Pouly15thesis,FLBP17cmsb}
consists in introducing a variable $X$, initialized to value $x$,
multiplying the terms of the PIVP for generating the function,
and following a decreasing exponential to halt the PIVP on the prescribed input.
The previous PIVP of degree 3 thus becomes a PIVP of degree 4 by this transformation:
\begin{align*}
	\frac{dX}{dt} &= -X& \frac{dA}{dt} &= 2.T.B.X& \frac{dB}{dt} &= -2.T.B^2.X& \frac{dT}{dt} &= X\\
	X(0)&=x&A(0)&=1& B(0)&=1 & T(0)&=0
\end{align*}
The quadratization algorithm~\cite{HFS20cmsb} now introduces new variables $BX$ and $TBX$ for the corresponding monomials, and removes
variables $T$ and $B$. This generates the following quadratic PIVP\@:
\begin{align*}
	\frac{dX}{dt} &= -X &\frac{dBX}{dt} &= \frac{dB}{dt}X + B\frac{dX}{dt} = -2.BX.TBX-BX &\\
	\frac{dA}{dt} &= 2.TBX &\frac{dTBX}{dt} &= BX.X -2.TBX^2 -TBX &&\\
	X(0)&=x &\quad BX(0)&=x \quad A(0)=1 \quad TBX(0)=0&
\end{align*}
% with the command\linebreak
% {\small\verb|compile_function(A=log(1+X^2), X, quadratic_reduction:sat_species)|}:
and finally the following elementary CRN which computes the $\log(1+x^2)$ function:
\begin{align*}
\mathit{BX}+\mathit{TBX}&\xrightarrow{2}\mathit{TBX}& X&\xrightarrow{1}\emptyset& \mathit{BX}&\xrightarrow{1}\emptyset\\
\mathit{BX}+X&\xrightarrow{1}\mathit{BX}+\mathit{TBX}+X& \mathit{TBX}&\xrightarrow{1}\emptyset\\
2.\mathit{TBX}&\xrightarrow{2}\mathit{TBX}& \mathit{TBX}&\xrightarrow{2}A+\mathit{TBX}\\
X(0)&=x&\mathit{BX}(0)&=x&A(0)&=1\\
\end{align*}

\subsection{Elementary Functions as Compilation Pipeline Input Language}

In mathematics, elementary functions refer to unary functions (over the reals
in our case) that are defined as a sum, product or composition of finitely many
polynomial, rational, trigonometric, hyperbolic, exponential functions and their inverses.
Most of these functions are defined on the real axis but a few exceptions are worth
mentioning: the inverse of some function is restricted to the image of \R\ by the function
(e.g., $\arccos$ is only defined on $[-1,1]$) and the exponentiation may
be non-analytic in $0$ and is thus considered elementary only on an open
interval that does not include $0$.

The set of elementary functions of $x$ is formally defined as the least set of functions containing:
\begin{itemize}
\item Constants: $2,\ \pi ,\ e$, etc.
\item Polynomials of $x:  x+1,\ x^{2},\ x^{3}-42.x$, etc.
\item Powers of $x: {\displaystyle \sqrt {x},\ \sqrt[{3}]{x}, x^{-4}}$, etc.
\item Exponential and logarithm functions: ${\displaystyle e^{x}, \ln x}$
\item Trigonometric functions: $\sin x,\ \cos x,\ \tan x$, etc.
\item Inverse trigonometric functions: $\arcsin x,\ \arccos x$, etc.
\item Hyperbolic functions: $\sinh x,\ \cosh x$, etc.
\item Inverse hyperbolic functions: $\operatorname{arsinh} x,\ \operatorname{arcosh} x$, etc.
\end{itemize}
and closed by arithmetic operations (addition, subtraction, multiplication, division)
and composition.
Elementary functions are also closed by differentiation but not necessarily by integration.
On the other hand hyper-geometric functions, Bessel functions, gamma, zeta functions, are examples
of (computable) non-elementary functions.

\section{Polynomialization Algorithm for Elementary ODEs}\label{algo}

\subsection{Polynomialization Algorithm}\label{sub:algo}

The core of Alg.~\ref{algo:podeize} for polynomializing an EODE system
is the detection of the elements of the derivatives that are not polynomial and their introduction
as new variables. Then symbolic derivation and syntactic substitution allow us to compute the
derivatives of the new variables and to modify the system of equations accordingly.

It is worth noting that the list of substitutions has to be memorized along the way, therefore handling
an algebraic-differential system during the execution of the algorithm, since they may reappear during
the derivation step. This typically occurs when the derivation graph harbors a cycle like:
$\cos \rightarrow \sin \rightarrow \cos$ (Fig.~\ref{fig:deriv}).

Nevertheless, a particular treatment has to be applied to the case of non-integer or
negative power as they form an infinite set and may thus produce infinite chains of
derivations. This can be seen if we try to apply naively
Alg.~\ref{algo:podeize} on this simple example:
\[\frac{dA}{dt} = A^{0.4}\]
for which we introduce the new variable $B = A^{0.4}$ with
\[\frac{dB}{dt} = 0.4 A^{-0.6} \qquad\frac{dA}{dt} = 0.4 A^{-0.2}\]
At this point, it is tempting to introduce $C = A^{-0.2}$ but that would lead to an
infinite loop, introducing more and more powers of $A$.
\begin{algorithm}
\caption{Polynomialization of an EODE system}\label{algo:podeize}
\begin{algorithmic}[1]
	\State \textbf{Input}: A set of ODEs of the form $\{x' = f_x(x, y, \ldots), y'=f_y(x, y	\ldots), \ldots\}$.
	\State \textbf{Output}: A set of PODEs where the initial variables $x, y,
	\ldots$ are still present and accept the same solutions.

	%\Function{polynomialize\_ode}{$ODE$}
	%	\State $Variables \gets$ the set of variables of $ODE$
		\State \textbf{Initialize} $Transformations \gets \emptyset$ and $PolyODE \gets \emptyset$
		\While{ODE is not empty}
			\State take and remove $Var' = Derivative$ from $ODE$;
			\State $NewDerivative \gets$ apply $Transformations$ to $Derivative$;
			\State $Terms \gets$ set of maximal non-polynomial subterms of $NewDerivative$; % of numeric constants, parameters, variables of $ODE$;
			\ForAll{$Term$ in $Terms$}
			%	\State $NewName \gets$ new variable name
				\State add $(Term \mapsto NewVar)$ to $Transformations$;
			%	\State add $NewName$ to $Var(ODE)$
				\State $TermDerivative \gets$ the symbolic derivative of $Term$;
				\State add $(NewVar' = TermDerivative)$ to $ODE$;
			\EndFor
			\State $PolyDerivative \gets$ apply $Transformations$ to $Derivatives$;
			\State add $(Var' = PolyDerivative)$ to $PolyODE$;
		\EndWhile
		\State \textbf{return} $PolyODE$
	%\EndFunction
\end{algorithmic}
\end{algorithm}

This can be easily avoided by introducing $C = \frac{1}{A}$ instead, then:
\[\frac{dB}{dt} = 0.4 B^2 C\]
\[\frac{dC}{dt} = -\frac{1}{A^2} \qquad\frac{dA}{dt} = -C^2 B\]
There is therefore a specific treatment to do when introducing the new variable of an
exponentiation in order to force the
algorithm to use the inverse variable instead of an infinite sequence of variables. For this,
when adding the new variable $N = X^p$ to the system, we explicitly replace the expression
$X^{p-1}$ by $N/X$ in the derivatives, thus making the use of $1/X$ a natural consequence.
Of course, this makes the final PODE non analytic in $X=0$. This is
linked to the fact that exponentiation apart from the
polynomial case is actually not analytic in $0$, and it is thus not
surprising that computation fails if we reach a time where $X=0$.

\subsection{Interval of definition}

Elementary functions are analytic upon open interval of their domain, but may suffer from
non-analyticity on the boundary. For example exponentiation with a non integer coefficient
may be extended by continuity in $0$ but is not analytic here. During our
polynomialization, this kind of behaviour may lead to the appearance of species that
diverge on these points.  This is important as it can be shown that only analytic
functions can be generated by a  PIVP\@.

In particular, the absolute value function over the reals is elementary as it can be
expressed as the composition of a power and root of $x: \lvert x\rvert={\sqrt {x^{2}}}$.
But is not analytic on $0$.
% For instance, with the example of the absolute value given above, $|x| = \sqrt{x^2}$,
Hence, if we consider the EODE $\frac{dx}{dt} = |x|$, our polynomialization will introduce
the variables $y=|x|$ and $z = \frac{1}{x}$, to obtain the PODE\@: \[\frac{dx}{dt}=
y\qquad\frac{dy}{dt} = y^2.z\qquad \frac{dz}{dt} =  - (z^2.y)\] And when $x$ approaches
$0$, the variable $z$, its inverse, will diverge.

The unicity of the solution of a PIVP is constrained by the initial conditions, but this
unicity is not ensured when passing a non-analyticity. A consequence of this remark is
that our compiled solution is defined from its initial condition ($t=0$) up to the first
non-analyticity of the compiled function.

\subsection{Termination}

Before proving the termination of our algorithm for the input set of elementary functions over the reals, we
can show a general lemma for any set $F$ of formally differentiable, possibly multivariate, real functions.
Let $\overline{F}$ denote the closure of $F \cup \R$ by addition and multiplicationi, that
is the algebra of $F$ over $\R$.

\begin{lemma}\label{lmm:finite}
  For any finite set $F$ of formally differentiable functions over the reals such that $\forall f \in F, f' \in \overline{F}$,
  Alg.~\ref{algo:podeize} terminates.
\end{lemma}
\begin{proof}
In the \textbf{while} loop, since we detect all the non-polynomial parts of the derivatives in one pass, we are sure
that the derivative at hand in each \textbf{while} step becomes polynomial.
New variables may however be introduced in such a step, thus the only possibility of non-termination
is to introduce an infinity of new variables.

The proof of termination proceeds by cases on the structure of the derivatives.
%As explained above, we only have to prove that the derivatives of the new variables themselves are either polynomial, or at
%least will not introduce an infinite succession of variables.
Suppose we have an ODE on the set of variables $\{X_i\}$
with $i \in \left[1, n\right]$ and let us denote by $d(X_i)$ the derivative of $X_i$.
First, it is obvious that if $d(X_i)$ is a single variable or a constant (numeric or
parameter) then we have nothing to do. Similarly, for every expression composed by
addition or multiplication of such terms, they are already polynomial.

If the new variable is a function of several variables (here $2$): $Y = f(X_j, X_k)$, we
have:
\[
	\frac{dY}{dt} = \frac{\partial f(X_j, X_k)}{\partial X_j}d(X_j) +  \frac{\partial f(X_j, X_k)}{\partial X_k}d(X_k)
\]
and as the addition is allowed in polynomial, we can consider separately the two
derivatives. Thus, we can restrict ourselves to the case of functions of a single variable.

For a function of a single variable with no composition we have:
\[
	\frac{dy}{dt} = \frac{df(X_i)}{dt} = f'(X_i) d(X_i)
\]
we have seen that $d(X_i)$ is already polynomial but nothing ensures that $f'$ is. However
	by definition of our set $F$, $f'(X_i)$ may be expressed with other function of $X$ in
	the set $F$, all applied to $X_i$, since composition is not allowed in the construction
	of $\overline{F}$. As the set $F$ is finite, we are sure to terminate after introducing
	at most $|F|$ variables. It is thus important
	to not include the closure by composition in the definition of our set $\overline{F}$.
	Indeed a function such that its derivative would be of the form: $f'(x) = f(f(x))$ may
	lead to an infinite loop for our algorithm.

	Finally, when facing a composition, e.g. $f(g(X))$, we replace it by a new variable, say $y$, with the standard derivation rule:
\[
	\frac{dy}{dt} = \frac{df(g(X_i))}{dt} = f'(g(X_i)) g'(X_i) d(X_i)
\]
	We thus have two different chains of variables to introduce: at first $f(g(x))$, $f'(g(x))$,
	$f''(g(x))$, etc.\ and in a second time: $g'(x)$, $g''(x)$. The important point to
	remark is that there is no mixing: all derivatives of $f$ are applied to $g(x)$ and
	neither to the derivatives of $g$.
	By the same argument as for the case without composition, the polynomialization of both
$f'(g(X_i))$ and $g'(X_i)$ terminates.
\end{proof}
\begin{figure}
\begin{center}
	\includegraphics[width=0.85\textwidth]{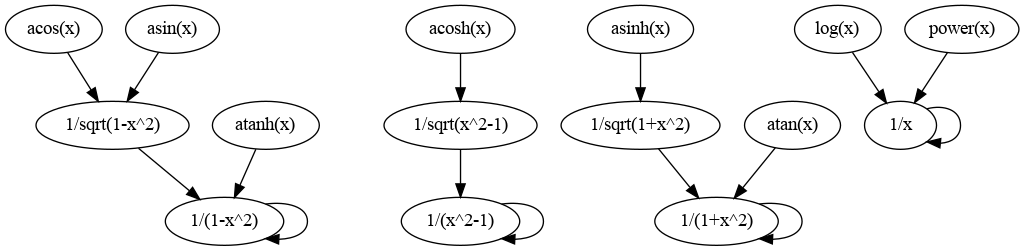}\\
	\includegraphics[width=0.75\textwidth]{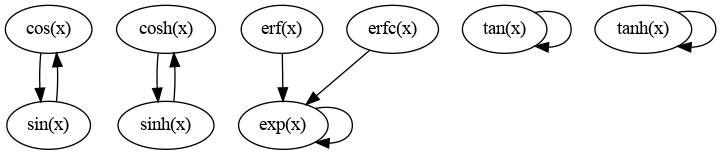}\caption{Dependency graph of the derivatives for the elementary functions.
	Each function is given up to a polynomial composition, hence the
	derivative of $\arccos{x}$ which is $\frac{1}{\sqrt{1-x^2}}$ points to the
	$\frac{1}{\sqrt{x}}$ node. Note that the out-degree of each node is $1$ as the
	derivative of each of these functions is always a single function or the
	composition of a polynomial and a single function.}~\label{fig:deriv}
\end{center}
\end{figure}
\begin{corollary}
	Alg.~\ref{algo:podeize} terminates on elementary functions.
\end{corollary}
\begin{proof}
  Let $F$ be the set of elementary functions. %\footnote{The error function, while	not derivable is usually present in most arithmetic package and has been added in our implementation	for completude.}.
	Fig.~\ref{fig:deriv} displays the dependency graph of the set $F$ and its derivatives. 
        Apart from the exponentiation, we can see
	that this set obeys to the condition of Lemma~\ref{lmm:finite}.

	The only difficulty thus comes from the exponentiation as it actually defines an infinite set
	of functions of the form $f(x) = x^\alpha$ for any real constant $\alpha$. Nevertheless, as explained
	above, the algorithm terminates also in this case since we express the derivative in the form:
	\begin{equation}
		f'(x) = \alpha x^{\alpha-1} = \alpha f(x) \frac{1}{x}
	\end{equation}
	and the inverse function is in $F$, henceforth the case of exponentation also terminates.
	%Combining these two remarks concludes the proof.
\end{proof}

\subsection{Complexity}

To estimate the computational time complexity of Alg.~\ref{algo:podeize}, let us first determine a bound on the number
of new variables introduced in the system.
\begin{proposition}
Alg.~\ref{algo:podeize} introduces at most a linear number of variables in the size of the input.
  \end{proposition}
\begin{proof}
For each elementary function $f$, represented in Fig.~\ref{fig:deriv}, we have to
introduce recursively all the functions on which rely its derivative, in terms of the
(directed) dependency graph, this is the reachable set starting from $f$. In this case
the largest reachable set is of cardinal $3$. More generally, let us call $\ell$ the
cardinal of the largest set of reachable nodes starting form a single node.  Then, each
elementary function that is not already polynomial introduces at most $\ell$ variables.

Hence, we have to introduce at most $\ell F$ variables where $F$ is the number of
functions used in the ODE\@. Of course, $F$ is bounded by the size of the input.
  \end{proof}
\begin{proposition}
Alg.~\ref{algo:podeize} has a quadratic time complexity.
  \end{proposition}
\begin{proof}
To introduce a new variable, we have to first compute its derivative, and then substitute
its expression in the reminder of the ODE\@. Both operations are
% done in parallel and are
linear in the size of the current system size, and we just have seen that this one will grow
	at most as a linear function of the input size, giving only a linear dependency.
Thus, our algorithm is quadratic in the size of the input.
\end{proof}

\subsection{Remark on the Compilation of the Exponentiation}\label{sec:power}

In the compilation of the exponentiation described in Sec.~\ref{sub:algo} we introduce two variables,
but this raises an interesting question concerning the conditions under which a
differential equation of the form:
$\displaystyle\frac{dA}{dt} = A^\alpha$
can be set in polynomial form by introducing only one variable.
Indeed, if we introduce $B = A^\beta$ we have:
\[\frac{dB}{dt} = \beta A^{\beta-1}\qquad \frac{dA}{dt} = \beta A^{\alpha+\beta-1}\]
and to be polynomial we need to find four positive integers $i,j,k,l$ such that:
\begin{align*}
	\alpha &= i + j \beta\\
	\alpha + \beta - 1 &= k + l \beta
\end{align*}
We thus need that $\alpha$ be of the form:
$\displaystyle\alpha = \frac{j(k+1)+i(1-l)}{j+1-l}$

Clearly, only a fractional power can be reduced in one step. Now suppose that
$\alpha =\frac{p}{q}$ then we have by identifying numerator and denominator and setting $l=2$:
\[j = q-1+l = q+1 \qquad\text{and}\qquad k = \frac{p+i}{1+q}-1\]
an equation that always admits solutions with $i$ and $k$ positive.
For example $\alpha =
\frac{1}{3}$ may be solved with $\beta = \frac{-2}{3}$ and $i=3,j=4,k=0,l=2$. But this uses
a polynomial of order $7$ which may impact negatively the final quadratization phase~\cite{HFS20cmsb}.
For these reasons, we chose to treat exponentiation by systematically introducing two variables.

%for the sake of simplicity. Indeed, that solution is
%obvious and does not need to estimate the best choice for the value of $\beta$ in
%order to introduce polynomial of the lowest possible degree.

%Furthermore, we do not want to rely on a technique that is only valid for
%a particular equation, and as soon as a second monomial is added it would not be possible to
%use it anymore.

\section{CRN Compilation Pipeline for Elementary Functions}\label{pipeline}

\subsection{Detailed Example}

To illustrate the behavior of our complete pipeline schematized in Fig.~\ref{fig:pipeline},
let us consider the compilation of the Hill function $H = \frac{x}{1+x}$ 
into a CRN that computes $H(x)$ in the sense of Def.~\ref{computes},
i.e., such that the final concentration of species $H$ gives the result $\lim_{t
\rightarrow \infty} H(t) = \frac{x}{1+x}$.

The first step of the front-end transformation is to introduce a pseudo-time variable: $T=t$ and replace the input $X$ by $T$ in the
expression $H = \frac{T}{1+T}$, and then to compute the formal derivative of $H$ to obtain the ODE\@:
\[
	\frac{dH}{dt} = \frac{1}{{(1+T)}^2}\qquad \frac{dT}{dt} = 1
\]

The second step is to polynomialize that ODE with Alg.~\ref{algo:podeize}.
This is done by examining the right hand
side of the derivative of $H$ and introduce as new variable $A = \frac{1}{{(1+T)}}$ which
allows us to rewrite $\frac{dH}{dt} = A^2$. We also need to compute the derivative of this new
variable:
\[
	\frac{dA}{dt} = -A^2,\ \ \ \ A(0)=1
\]
This PIVP generates the time function $H(t)$ in the sense of Def.~\ref{generates}.

Now, we enter the compilation pipeline from PIVP described in~\cite{FLBP17cmsb}.
To compute the function $H(x)$ of input $x$, a new variable $X$ that obeys a decreasing
exponential is introduced and used to halt all other derivatives at time $x$:
\begin{align*}
	\frac{dH}{dt} &= A^2.X&
	\frac{dX}{dt} &= -X\\
	\frac{dA}{dt} &= -A^2.X&
	\frac{dT}{dt} &= X\\
		  H(0)&=0&
		  X(0)&=x\\
		  A(0)&=1&
		  T(0)&=0
\end{align*}
Then the quadratization step~\cite{HFS20cmsb}, introduces the new variable $B = A.X$. Interestingly, 
intermediary variables $X$ and $T$ are no longer used and removed. We finally get:
\begin{align*}
	\frac{dH}{dt} &= A.B &
	\frac{dB}{dt} &= -B -B^2 &
	\frac{dA}{dt} &= -A.B \\
		  H(0)&=0&
		  B(0)&=x&
		  A(0)&=1
\end{align*}
And the compiled CRN is:
\begin{align*}
	B &\xrightarrow{1} \emptyset &
	A+B &\xrightarrow{1} B+H&
	2.B &\xrightarrow{1} B\\
	A(0) &= 1 &
	B(0) &= \text{input}
\end{align*}

\subsection{Implementation}

Alg.~\ref{algo:podeize} is implemented by rewriting formal expressions using a simple algebraic normal form and standard derivation rules.

The most computationally expensive step of our complete compilation pipeline in Fig.~\ref{fig:pipeline}
is the quadratization of the intermediate PIVP to a PIVP of order at most $2$. This step is necessary to restrict ourselves
to elementary reactions with at most two reactants which are more amenable to real implementations with real enzymes.
While the existence of a quadratic form for any PODE can be simply shown by introducing an exponential number of variables~\cite{CPSW05ejde},
the problem of minimizing the dimension of that quadratization is NP-hard~\cite{HFS20cmsb}.
In the implementaiton of BIOCHAM used in the next section, we use both the MAXSAT algorithm described in~\cite{HFS20cmsb} (option \verb$sat_species$ below)
and a heuristic algorithm (option \verb$fastnSAT$ below)
to first obtain a subset of variables guaranteed to contain a quadratic solution,
and then call the MAXSAT solver (RC2\footnote{\url{https://pysathq.github.io/docs/html/api/examples/rc2.html}})
to minimize the dimension in that quadratization.
%Recently, a more general dimension minimization algorithm has been proposed in~\cite{BP21iwoca}.

\section{Evaluation}\label{eval}

In this section, we consider the benchmark of functions already considered in~\cite{HFS20cmsb} for quadratization problems.
Table~\ref{tab:bench} gives some performance figures about the complete compilation pipeline
in terms of total computation time and size of the synthesized CRNs,
with the two options discussed above for quadratization.
\begin{table}
\centering
	\begin{tabular}[pos]{lrrr@{\hskip 10pt}rrr}
		\toprule
			   & \multicolumn{3}{c}{\texttt{fastnSAT}} & \multicolumn{3}{c}{\texttt{sat\_species}} \\
    \cmidrule(lr){2-4}\cmidrule(lr){5-7}
    Function & time & number of  & number of & time & number of  & number of \\
    & (ms) & species & reactions & (ms) & species & reactions \\
    \midrule
    Hill1 & 80 & 4 & 5 & 85 & 3 & 3 \\
    Hill2 & 90 & 6 & 10 & 82 & 5 & 8 \\
    Hill3 & 100 & 6 & 10 & 115 & 6 & 12 \\
    Hill4 & 100 & 7 & 13 & 162 & 7 & 13 \\
    Hill5 & 110 & 8 & 16 & 550 & 7 & 11\\
    Hill10 & 160 & 13 & 31 & \multicolumn{3}{l}{timeout}\\
    Hill20 & 380 & 23 & 61 & \multicolumn{3}{l}{timeout}\\
    Logistic & 80 & 3 & 5 & 85 & 3 & 5\\
    Double exp. & 80 & 3 & 4 & 85 & 3 & 4\\
    Gaussian & 85 & 3 & 4 & 85 & 3 & 4\\
    Logit & 95 & 4 & 7 & 100 & 4 & 6\\
    \bottomrule
	\end{tabular}
	\caption{Performance results on the benchmark of CRN design problems of~\cite{HFS20cmsb}
          in terms of total compilation time, and size of the synthesized CRN with two options for quadratization.}\label{tab:bench}
\end{table}
It is worth noting that those synthetic CRNs are not unique and that other CRNs could be synthesized
for the same function, by making different choices in both our polynomialization and quadratization algorithms.
Even when imposing optimality in the number of introduced variables, there may exist several optimal CRNs. For example,
the two CRNs obtained for Hill4 compiled with the two options for quadratization are
different but both have the same number of species and reactions (Table~\ref{tab:bench}).

Let us examine the influence graphs of those synthetic CRNs since they provide a more compact abstract representation of the reaction graph~\cite{FS08tcs}.
Fig.~\ref{fig:influence} depicts the influence graphs between molecular species of the synthesized CRNs for
the hill
functions of order $3$ (Fig.~\ref{fig:influence}A.) and $5$
(Fig.~\ref{fig:influence}B.), the logistic function
(Fig.~\ref{fig:influence}C) and the square of the cosine function (Fig.~\ref{fig:influence}).
One can remark on those examples that the outputs of the synthesized CRN do not participate in any feedback reaction.
This is however not necessarily the case of the CRNs synthesized by our pipeline,
as shown for instance by the cosine function~\cite{FLBP17cmsb}.
\begin{figure}[h]
\begin{center}
	\includegraphics[width=0.8\textwidth]{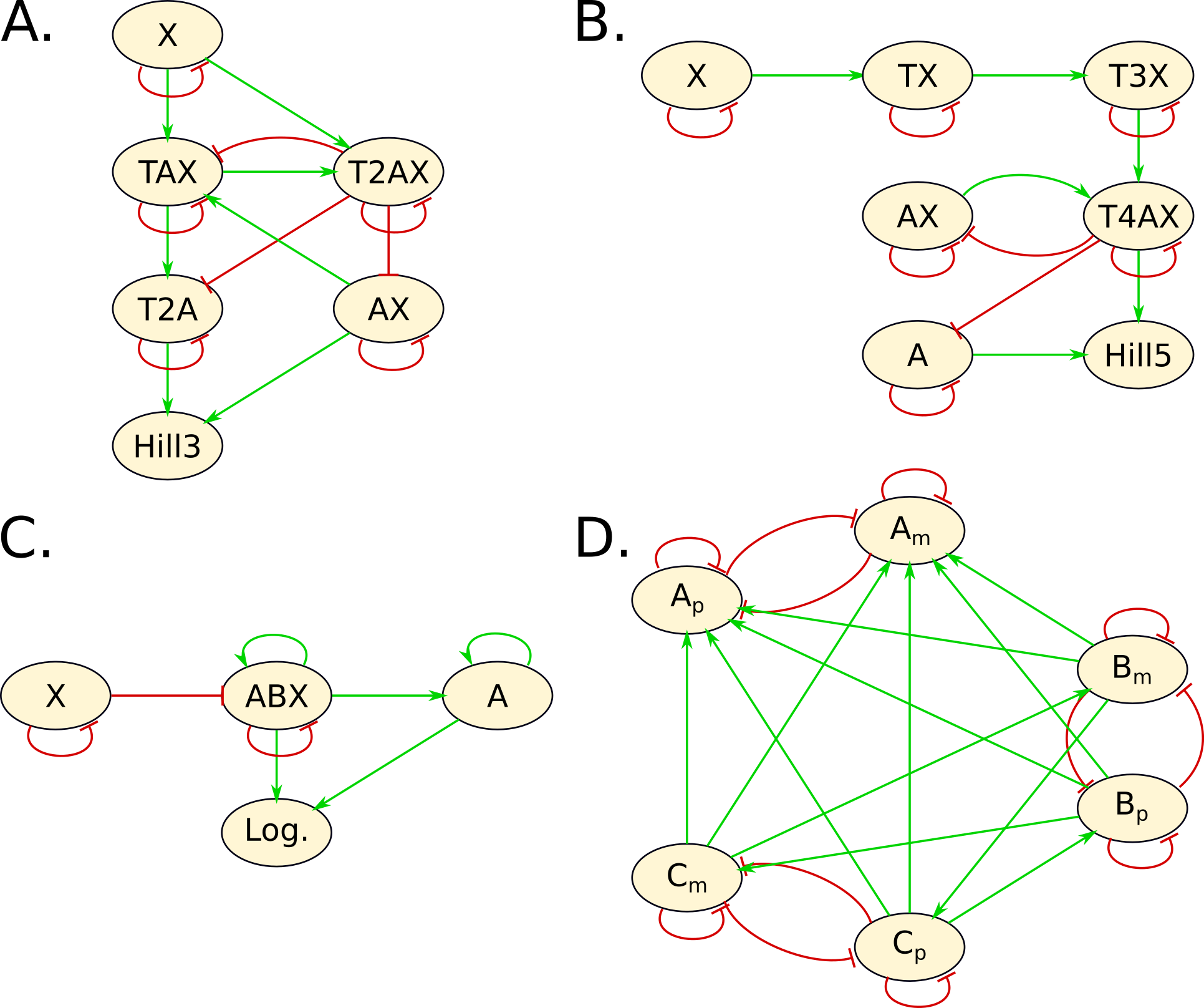}\\
	\caption{Influence graphs of four of the synthetic CRNs of Table~\ref{tab:bench}. \textbf{A.} and \textbf{B.}
	respectively implement the Hill function of order $3$ and $5$. \textbf{C.} corresponds
	to the logistic function and \textbf{D.} computes the square of the cosine of the time.
	In this last example, the output is only read on $A_p$ but the presence of its negative
	part $A_m$ is actually a crucial part of the computation despite having an essentially
	null concentration.}\label{fig:influence}
\end{center}
\end{figure}
More precisely, the Hill3 CRN in Table~\ref{tab:bench} synthesized with the \verb!sat_species! option is the following:
%\begin{figure}
	\begin{align*}
		\mathit{x} &\xrightarrow{1} \emptyset&
		\mathit{Ax} &\xrightarrow{1} \emptyset\\
		\mathit{TAx} &\xrightarrow{1} \emptyset&
		\mathit{T2Ax} &\xrightarrow{1} \emptyset\\
		\mathit{Ax} + \mathit{T2Ax} &\xrightarrow{3} \mathit{T2Ax}&
		\mathit{Ax} + \mathit{x} &\xrightarrow{1} \mathit{Ax} + \mathit{TAx} + \mathit{x} \\
		\mathit{T2Ax} + \mathit{TAx} &\xrightarrow{3} \mathit{T2Ax}&
		\mathit{TAx} &\xrightarrow{2} \mathit{T2A} + \mathit{TAx} \\
		\mathit{T2A} + \mathit{T2Ax} &\xrightarrow{3} \mathit{T2Ax}&
		\mathit{TAx} + \mathit{x} &\xrightarrow{2} \mathit{T2Ax} + \mathit{TAx} + \mathit{x} \\
		2.\mathit{T2Ax} &\xrightarrow{3} \mathit{T2Ax}&
		\mathit{Ax} + \mathit{T2A} &\xrightarrow{2} \mathit{Ax} + \mathit{T2A} + hill3 \\
		\mathit{x}(0) &= \text{input}&
		\mathit{Ax}(0) &= \text{input}
	\end{align*}
%\end{figure}
The Hill5 CRN in the table synthesized with the same option is:
	\begin{align*}
		\mathit{x} &\xrightarrow{1} \emptyset&
		\mathit{Ax} &\xrightarrow{1} \emptyset\\
		\mathit{Tx} &\xrightarrow{1} \emptyset&
		\mathit{T3x} &\xrightarrow{1} \emptyset\\
		\mathit{T4Ax} &\xrightarrow{1} \emptyset&
		\mathit{A} + \mathit{T4Ax} &\xrightarrow{5} \mathit{T4Ax} + hill5\\
		\mathit{Ax} + \mathit{T4Ax} &\xrightarrow{5} \mathit{T4Ax}&
		2.\mathit{x} &\xrightarrow{1} \mathit{Tx} + 2.\mathit{x}\\
		2.\mathit{Tx} &\xrightarrow{3} \mathit{T3x} + 2.\mathit{Tx}&
		\mathit{Ax}+\mathit{T3x} &\xrightarrow{4} \mathit{Ax}+\mathit{T3x}+\mathit{T4Ax}\\
		2.\mathit{T4Ax} &\xrightarrow{5} \mathit{T4Ax}&
		x(0)&=\text{input} \\
		\mathit{A}(0)&=1&
		\mathit{Ax}(0)&=\text{input}
	\end{align*}

From our computational point of view, the Hill5 CRN above is one synthetic analog of the natural MAPK CRN among others.
Indeed, the natural MAPK signalling network
has been shown in~\cite{HF96pnas}  to compute an ultrasensitive input/output function
which is well approximated by a Hill function of order 4.9.
Both the natural MAPK CRN and the synthetic Hill5 CRN thus compute a similar input/ouput function and it makes sense to try to compare their structure.
In term of size, the MAPK model of~\cite{HF96pnas}  comprises 22 species and 30 reactions,
while the Hill5 CRN synthesized by our pipeline now uses only $7$ formal molecular species and $11$ elementary reactions.
This shows a huge improvement with respect to our first results reported in~\cite{FLBP17cmsb}
where several tens of reactions were synthesized for Hill functions.
In term of topological structure,
we have checked that there exists (several) subgraph epimorphisms~\cite{GSF10bi} mapping the MAPK CRN to that Hill5 CRN,
meaning that the MAPK CRN somehow contains the core structure of the synthetic Hill5 CRN in some non-trivial sense.
The biological significance of those relationships is however still unclear,
although one could expect to explain it in terms of robustness properties~\cite{Cardelli14bmc}.
It is also worth noting that in the natural MAPK CRN structure,
the input is a catalyst that is not consumed by the downward reactions,
whereas in our CRN synthesis scheme, the input is generally consumed by the downward reactions.
The MAPK network  thus illustrates a case of online analog computation
which is currently not treated by our theoretical framework.

\section{Conclusion and Perspectives}

We have presented an algorithm to transform any system of elementary ordinary
differential equations to a polynomial ordinary differential equation system 
preserving the solutions of the original variables.
This algorithm of  quadratic time complexity introduces at
most a linear number of new variables.
This algorithm allows us to automatically compile any elementary mathematical function
into a finite CRN, using a pipeline of transformations starting from the formal derivation
of the elementary function to generate or compute, the polynomialization of the elementary
ODE, and continuing with the previous pipeline of~\cite{FLBP17cmsb} for the dual-rail
encoding of negative values of the PODE~\cite{HT79cmsjb,OK11iet}, the quadratization of the PODE~\cite{HFS20cmsb}, and the synthesis
of elementary reactions for the quadratic ODEs~\cite{FGS15tcs}.

The implementation in BIOCHAM-4 of this complete pipeline has been used to illustrate the CRNs
synthesized for a variety of elementary mathematical functions used as specification.  In
particular, the CRN synthesized for the Hill function of order 5 provides a synthetic
analog of the MAPK signalling network which has been shown to compute a similar
ultrasensitive input/output function~\cite{HF96pnas}.
In this compilation process, the quadratization part is the most complex one since
minimizing the dimension of the result is a NP-hard problem~\cite{HFS20cmsb}. 
On the benchmark presented here, our maxSAT implementation is sufficient but we also use
by default a heuristic algorithm that trades optimality for better performance. It
should also be noted that a new algorithm has been recently proposed for the global optimization
problem in~\cite{BP21iwoca}.

This work may be improved in several directions.
We might extend this approach to multivariate functions. This is not trivial as the trick of
halting the time at the input value needs be generalized to several inputs. 

Another important point is to investigate is the variety of different CRNs that can be synthesized
by our pipeline. As pointed out earlier, there may be several optimal solutions to the
quadratization problem and there may similarly be several polynomializations of a given ODE
introducing the same number of variables. Our pipeline make choices to deterministically propose one
solution, but being able to explore the set of solutions and compare their properties with
respect to metrics like robustness to initial conditions or reaction rates, or
imposing some similarity requirement with a given biological solution would be interesting.

Furthermore, the comparison to the MAPK network also points to the interesting
class of online computation  which does not consume the inputs, whereas in our approach the
input species are consumed, and the synthesized CRN would need to be reinitialized for
another computation.  This limitation is not a problem for one-shot CRN programs such as
those designed for medical diagnosis applications~\cite{CAFRM18msb}, but synthesizing a
CRN for computing an input/output function online appears to be a harder problem worthy
of further theoretical investigation.

\subsubsection*{Acknowledgment}
We acknowledge fruitful discussions with Olivier Bournez, François Lemaire, Gleb Pogudin and Amaury Pouly. This work was supported by
ANR-DFG SYMBIONT ``Symbolic Methods for Biological Networks'' project grant ANR-17-CE40-0036, and
ANR DIFFERENCE ``Complexity theory with discrete ODEs'' project grant ANR-20-CE48-0002.

\bibliographystyle{plain}
\bibliography{contraintes}

\begin{thebibliography}{10}

\bibitem{BFS18jtb}
Adrien Baudier, Fran{\c{c}}ois Fages, and Sylvain Soliman.
\newblock Graphical requirements for multistationarity in reaction networks and
  their verification in biomodels.
\newblock {\em Journal of Theoretical Biology}, 459:79--89, December 2018.

\bibitem{BCGH06complexity}
Olivier Bournez, Manuel~L. Campagnolo, Daniel~S. Gra{\c c}a, and Emmanuel
  Hainry.
\newblock {Polynomial differential equations compute all real computable
  functions on computable compact intervals}.
\newblock {\em {Journal of Complexity}}, 23(3):317--335, 2007.

\bibitem{BP21iwoca}
Andrey Bychkov and Gleb Pogudin.
\newblock Optimal monomial quadratization for ode systems.
\newblock In {\em Proceedings of the IWOCA 2021 - 32nd International Workshop
  on Combinatorial Algorithms}, July 2021.

\bibitem{Cardelli14bmc}
Luca Cardelli.
\newblock Morphisms of reaction networks that couple structure to function.
\newblock {\em BMC systems biology}, 8(84), 2014.

\bibitem{CPSW05ejde}
David~C. Carothers, G.~Edgar Parker, James~S. Sochacki, and Paul~G. Warne.
\newblock Some properties of solutions to polynomial systems of differential
  equations.
\newblock {\em Electronic Journal of Differential Equations}, 2005(40):1--17,
  2005.

\bibitem{CDS12nc}
Ho-Lin Chen, David Doty, and David Soloveichik.
\newblock Deterministic function computation with chemical reaction networks.
\newblock {\em Natural computing}, 7433:25--42, 2012.

\bibitem{CSWB09ab}
Matthew Cook, David Soloveichik, Erik Winfree, and Jehoshua Bruck.
\newblock Programmability of chemical reaction networks.
\newblock In Anne Condon, David Harel, Joost~N. Kok, Arto Salomaa, and Erik
  Winfree, editors, {\em Algorithmic Bioprocesses}, pages 543--584. Springer
  Berlin Heidelberg, Berlin, Heidelberg, 2009.

\bibitem{CAFRM18msb}
Alexis Courbet, Patrick Amar, Fran{\c{c}}ois Fages, Eric Renard, and Franck
  Molina.
\newblock Computer-aided biochemical programming of synthetic microreactors as
  diagnostic devices.
\newblock {\em Molecular Systems Biology}, 14(4), 2018.

\bibitem{CF06siamjam}
Gheorghe Craciun and Martin Feinberg.
\newblock Multiple equilibria in complex chemical reaction networks: {II}. the
  species-reaction graph.
\newblock {\em SIAM Journal on Applied Mathematics}, 66(4):1321--1338, 2006.

\bibitem{DRSL13nature}
Ramiz Daniel, Jacob~R. Rubens, Rahul Sarpeshkar, and Timothy~K. Lu.
\newblock Synthetic analog computation in living cells.
\newblock {\em Nature}, 497(7451):619--623, 05 2013.

\bibitem{FLBP17cmsb}
Fran\c{c}ois Fages, Guillaume Le~Guludec, Olivier Bournez, and Amaury Pouly.
\newblock {Strong Turing Completeness of Continuous Chemical Reaction Networks
  and Compilation of Mixed Analog-Digital Programs}.
\newblock In {\em {CMSB'17}: Proceedings of the fiveteen international
  conference on Computational Methods in Systems Biology}, volume 10545 of {\em
  Lecture Notes in Computer Science}, pages 108--127. Springer-Verlag,
  September 2017.

\bibitem{FGS15tcs}
Fran{\c{c}}ois Fages, Steven Gay, and Sylvain Soliman.
\newblock Inferring reaction systems from ordinary differential equations.
\newblock {\em Theoretical Computer Science}, 599:64--78, September 2015.

\bibitem{FS08tcs}
Fran{\c{c}}ois Fages and Sylvain Soliman.
\newblock Abstract interpretation and types for systems biology.
\newblock {\em Theoretical Computer Science}, 403(1):52--70, 2008.

\bibitem{Feinberg77crt}
Martin Feinberg.
\newblock Mathematical aspects of mass action kinetics.
\newblock In L.~Lapidus and N.~R. Amundson, editors, {\em Chemical Reactor
  Theory: A Review}, chapter~1, pages 1--78. Prentice-Hall, 1977.

\bibitem{GSF10bi}
Steven Gay, Sylvain Soliman, and Fran{\c{c}}ois Fages.
\newblock A graphical method for reducing and relating models in systems
  biology.
\newblock {\em Bioinformatics}, 26(18):i575--i581, 2010.
\newblock special issue ECCB'10.

\bibitem{GC03}
D.S. Gra\c{c}a and J.F. Costa.
\newblock {Analog computers and recursive functions over the reals}.
\newblock {\em Journal of Complexity}, 19(5):644--664, 2003.

\bibitem{G11IEEE}
Chenjie Gu.
\newblock Qlmor: A projection-based nonlinear model order reduction approach
  using quadratic-linear representation of nonlinear systems.
\newblock {\em IEEE Transactions on Computer-Aided Design of Integrated
  Circuits and Systems}, 30(9):1307--1320, 2011.

\bibitem{HT79cmsjb}
V.~H{\'a}rs and J.~T{\'o}th.
\newblock On the inverse problem of reaction kinetics.
\newblock In M.~Farkas, editor, {\em Colloquia Mathematica Societatis J{\'a}nos
  Bolyai}, volume~30 of {\em Qualitative Theory of Differential Equations},
  pages 363--379, 1979.

\bibitem{HFS20cmsb}
Mathieu Hemery, Fran{\c{c}}ois Fages, and Sylvain Soliman.
\newblock On the complexity of quadratization for polynomial differential
  equations.
\newblock In {\em {CMSB'20}: Proceedings of the eighteenth international
  conference on Computational Methods in Systems Biology}, Lecture Notes in
  BioInformatics. Springer-Verlag, September 2020.

\bibitem{HF96pnas}
Chi-Ying Huang and James~E. Ferrell.
\newblock Ultrasensitivity in the mitogen-activated protein kinase cascade.
\newblock {\em PNAS}, 93(19):10078--10083, September 1996.

\bibitem{OK11iet}
K.~Oishi and E.~Klavins.
\newblock Biomolecular implementation of linear i/o systems.
\newblock {\em IET SYstems Biology}, 5(4):252--260, 2011.

\bibitem{Pouly15thesis}
Amaury Pouly.
\newblock {\em {Continuous models of computation: from computability to
  complexity}}.
\newblock PhD thesis, Ecole Polytechnique, July 2015.

\bibitem{RRD16bbs}
L.~Rizik, Y.~Ram, and R.~Danial.
\newblock Noise tolerance analysis for reliable analog and digital computation
  in living cells.
\newblock {\em J Bioengineer \& Biomedical Sci}, 6(186), 2016.

\bibitem{S12arXiv}
Sriram Sankaranarayanan.
\newblock Change-of-bases abstractions for non-linear systems.
\newblock {\em arXiv preprint arXiv:1204.4347}, 2012.

\bibitem{SK13nature}
Herbert~M. Sauro and Kyung Kim.
\newblock Synthetic biology: It's an analog world.
\newblock {\em Nature}, 497(7451):572--573, 05 2013.

\bibitem{Shannon41}
C.E. Shannon.
\newblock {Mathematical theory of the differential analyser}.
\newblock {\em Journal of Mathematics and Physics}, 20:337--354, 1941.

\end{thebibliography}

\end{document}